\documentclass[11pt]{article}

\usepackage{fullpage,latexsym,amsmath,amsfonts}

\usepackage{xcolor}
\usepackage[ruled]{algorithm2e}
\usepackage{hyperref}

\def\01{\{0,1\}}
\newcommand{\ceil}[1]{\lceil{#1}\rceil}

\newcommand{\ket}[1]{|#1\rangle}

\newcommand{\ketbra}[2]{|#1\rangle\langle#2|}
\newcommand{\Tr}{\mbox{\rm Tr}}

\newcommand{\E}{\mathbb{E}}
\newcommand{\Exp}{\mathbb{E}}
\newtheorem{definition}{Definition}
\newtheorem{theorem}{Theorem}

\newtheorem{conjecture}{Conjecture}
\newtheorem{fact}{Fact}

\newtheorem{remark}{Remark}
\newcommand{\I}{{I}}

\newcommand{\id}{\mathbb{I}}

\usepackage{mathtools}

\newcommand {\br} [1] {\ensuremath{ \left( #1 \right) }}

\newcommand {\minusspace} {\: \! \!}

\newcommand {\fn} [2] {\ensuremath{ #1 \minusspace \br{ #2 } }}

\newcommand{\defeq}{\ensuremath{ \stackrel{\mathrm{def}}{=} }}

\newcommand{\suppress}[1]{}

\newcommand{\dmax}[2]{\fn{{D}_{\max}}{#1 \middle\| #2}}
\newcommand {\imax}{\ensuremath{{I}_{\max}}}
\newcommand{\mI}{\imax}

\def\cH{\mathcal{H}}
\DeclareMathOperator{\supp}{supp}
\def\cX{\mathcal{X}}
\def\cY{\mathcal{Y}}
\def\nn{\nonumber}

\usepackage{hyperref}
\hypersetup{
    colorlinks,
    linkcolor={blue!100!black},
    citecolor={blue!100!black},
}

\newenvironment{proof}
{\noindent {\bf Proof. }}
{{\hfill $\Box$}\\
 \smallskip}

\begin{document}

\title{Getting almost all the bits from a quantum random access code}
\author{Han-Hsuan Lin\thanks{Department of Computer Science, National Tsing Hua University,  Hsinchu 30013, Taiwan. Supported by NSTC QC project under Grant no.\  111-2119-M-001-006-. {\tt linhh@cs.nthu.edu.tw}}
\and
Ronald de Wolf\thanks{QuSoft, CWI and University of Amsterdam, the Netherlands. Partially supported by the Dutch Research Council (NWO) through Gravitation-grant Quantum Software Consortium, 024.003.037. {\tt rdewolf@cwi.nl}}
}
\maketitle

\begin{abstract}
A quantum random access code (QRAC) is a map $x\mapsto\rho_x$ that encodes $n$-bit strings $x$ into $m$-qubit quantum states $\rho_x$, in a way that allows us to recover any one bit of $x$ with success probability $\geq p$. The measurement on $\rho_x$ that is used to recover, say, $x_1$ may destroy all the information about the other bits; this is in fact what happens in the well-known QRAC that encodes $n=2$ bits into $m=1$ qubits. Does this generalize to large $n$, i.e., could there exist QRACs that are so ``obfuscated'' that one cannot get much more than one bit out of them? Here we show that this is not the case: for every QRAC there exists a measurement that (with high probability) recovers the full $n$-bit string $x$ up to small Hamming distance, even for the worst-case~$x$.
\end{abstract}

\section{Introduction}

\subsection{Quantum Random Access Codes}

Quantum states of $m$ qubits are described by $2^m$ complex amplitudes, and it is tempting to try to use this exponential ``space'' to encode a lot of classical information in a quantum state. Holevo's theorem~\cite{holevo} says that in general we cannot use this exponential space: if we encode some classical random variable $X$ in an $m$-qubit quantum state $\rho_X$ (via a map $x\mapsto\rho_x$), and then do a measurement on $\rho_X$ to produce a classical random variable $Y$, then the mutual information between $X$ and $Y$ is at most $m$. So we cannot cram more than $m$ bits of information into $m$ qubits. 

\emph{Quantum random access} codes, first introduced in~\cite{ambainis:rac}, address a more subtle situation where we demand less from the decoding: here $x\in\01^n$ and we only want to be able to recover the bit $x_i$ of $x$ from $\rho_x$ with success probability~$p>1/2$, for any $i\in[n]$ of our choice. 
Note that this is different from recovering \emph{all} bits of $x$, since the measurement that recovers $x_1$ can change the $m$-qubit state, making it harder (or even impossible) to recover other bits of~$x$. More formally:

\begin{definition}
An $(n,m,p)$-quantum random access code (QRAC) is a map $x\mapsto\rho_x$ from the set of $n$-bit strings to the set of $m$-qubit mixed states,
such that for each $i\in[n]$ there exists
a 2-outcome POVM  $\{M_0,M_1\}$ such that for each $x\in\01^n$ and $b\in \01$
we have $\Tr(M_b\rho_x)\geq p$ if $x_i=b$.
\end{definition}

For example, we can encode $n=2$ classical bits into $m=1$ qubits by mapping
(for $b\in\01$) the 2-bit string $x=0b$ to the qubit $\cos(\pi/8)\ket{0}+(-1)^b\sin(\pi/8)\ket{1}$, and mapping $x=1b$ to $(-1)^b\sin(\pi/8)\ket{0}+\cos(\pi/8)\ket{1}$. See Figure~\ref{fig:2to1QRAC} for illustration.
A measurement in the computational basis recovers the first bit $x_1$ with success probability $\cos(\pi/8)^2\approx 0.85$, and a measurement in the Hadamard basis recovers $x_2$ with the same success probability.

\begin{figure}[hbt]
\centering
\setlength{\unitlength}{0.2mm}
\begin{picture}(270,160)
\put(208,45){$\ket{0}$}
\put(93,160){$\ket{1}$}
\put(100,50){\vector(1,3){32}}
\put(100,50){\vector(-1,3){32}}
\put(100,50){\vector(0,1){100}}
\put(100,50){\vector(1,0){100}}
\put(100,50){\vector(3,1){97}}
\put(100,50){\vector(3,-1){97}}
\put(200,85){$\rho_{00}$}
\put(200,15){$\rho_{01}$}
\put(50,160){$\rho_{11}$}
\put(130,160){$\rho_{10}$}
\end{picture}
\caption{The 4 pure states of the standard 2-to-1 QRAC}\label{fig:2to1QRAC}
\end{figure}
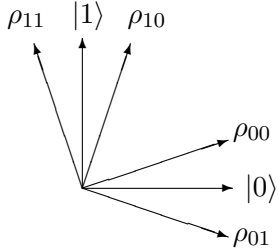

We know almost exactly how many qubits are required for an $(n,m,p)$-QRAC: Ambainis et al.~\cite{ambainis:rac} showed that $m\leq (1-H(p))n+O(\log n)$ qubits are \emph{sufficient}, even just using classical random variables (i.e., diagonal density matrices~$\rho_x$), and Nayak~\cite{nayak:qfa} proved that $m\geq (1-H(p))n$ qubits are \emph{necessary}, improving an $\Omega(n/\log n)$ lower bound from~\cite{ambainis:rac}.\footnote{Here $H(p)=-p\log_2(p)-(1-p)\log_2(1-p)$ denotes the binary entropy. The information-theoretic intuition behind Nayak's $m\geq (1-H(p))n$ lower bound is quite clear: $\rho_x$ allows us to predict $x_i$ with success probability $\geq p$, which reduces the entropy of $x_i$ from 1 bit to at most $H(p)$ bits, so gives at least $1-H(p)$ bits of information; this is true for all $n$ bits of $x$ simultaneously. The two papers \cite{ambainis:rac} and~\cite{nayak:qfa} were later combined in the journal version~\cite{ambainis:racj}.} 

QRACs have a number of  applications in quantum information, such as lower bounding the size of 1-way quantum finite automata~\cite{ambainis:rac,nayak:qfa,ambainis:racj}, 1-way quantum communication complexity~\cite{buhrman&wolf:qcclower,klauck:qpcom,gkrw:identificationj,brw:hypercontractive}, and the length of \emph{classical} locally decodable error-correcting codes~\cite{kerenidis&wolf:qldcj,brw:hypercontractive}.
They have also been used to show upper bounds, for instance that $n$-qubit density matrices can be learned (in a weak sense) from $O(n)$ measurement results~\cite{aaronson:qlearnability}.
They have been generalized to the setting of two parties (one holding $x$, the other holding $i$) who share randomness~\cite{ALMO:QRAC}, and to computing more general functions of up to $k$ of the bits of~$x$~\cite{brw:hypercontractive,DM:QRACf} rather than just recovering one bit~$x_i$.
We refer to~\cite[page 1--2]{DM:QRACf} for many more references to applications of QRACs and even experimental implementations.

\subsection{Our results: getting most of the bits of $x$ from $\rho_x$}

In a quantum random access code, the measurement needed to recover $x_1$ may significantly change the state, reducing the information we have about the other bits.
This is what happens for instance in the above 2-to-1 example (Figure~\ref{fig:2to1QRAC}): if we measure in the computational basis to recover $x_1$ and get the correct answer, then the state will have collapsed to $\ket{x_1}$, and all information about $x_2$ is lost; and if we got the wrong answer, then the post-measurement state $\ket{1-x_1}$ also has no information about $x_2$. Similarly, doing the Hadamard-basis measurement to recover $x_2$ will destroy all information about~$x_1$.
Can this ``obfuscating'' behavior be generalized? Maybe there are quantum random access codes that are so obfuscated that they allow to recover any one $x_i$ of our choice, but not many more bits? 
Such QRACs would be undesirable if one cares about extracting a lot of information about the whole $x$ (rather than just one bit $x_i$), but could be very desirable if one cares about hiding information for cryptographic purposes.

Our main result in this paper is a negative answer to this question: we show that every QRAC admits a measurement that recovers \emph{most} of the bits of $x$ correctly from one copy of $\rho_x$, in the sense of being able to obtain an $n$-bit string $y$ that has small Hamming distance with the encoded~$x$.

We prove this first in an average-case setting (Section~\ref{sec:avcase}), where we have some probability distribution $\{P_x\}$ on $x\in\01^n$, and then (in Section~\ref{sec:worstcase}) use the minimax theorem of game theory to extend the result to a worst-case setting, i.e., we obtain a measurement that works for every $x$.
The measurement that proves our result in the average-case setting is simple and well-known: it is the so-called Pretty Good Measurement (PGM) for identifying $x$ from $\rho_x$.
The PGM measures one copy of $\rho_x$ and produces a measurement outcome $y\in\01^n$.
The PGM cannot actually be very good at the task of identifying the whole $n$-bit $x$ exactly: another bound due to Nayak~\cite{nayak:qfa,ambainis:racj} says that the (uniform) average probability of fully identifying $x$ from $\rho_x$ (i.e., to have $y=x$) is at most $2^{m-n}$, which is tiny unless $m$ is almost as large as $n$.\footnote{\label{ft:Nayak}This result has a very elegant one-line proof: with POVM elements $\{Q_y\}$ (summing to identity on the $m$-qubit space), we have $\sum_x \Tr(Q_x\rho_x)\leq \sum_x \Tr(Q_x)=\Tr(\sum_x Q_x)=\Tr(I)=2^m$.} 
However, we prove here that for every possible QRAC and every distribution $\{P_x\}$, the PGM is actually very good at getting an individual $x_i$ correct (i.e., $x_i=y_i$ with high probability). 
Our treatment of the PGM and the induced one-bit PGM for $x_i$, is similar to how K\"onig and Terhal~\cite{konig2008bounded} show that the security of an extractor with one-bit output w.r.t.\ quantum side-information about~$x$ is roughly similar to its security w.r.t.\ classical side-information; similar techniques were also applied to one-way communication complexity under product distributions~\cite{boddu2023relating}.  
The above was all about recovering one bit $x_i$, but once we know that the PGM gets each bit $x_i$ correct with good probability, 
linearity of expectation then implies that (on average) the PGM actually gets \emph{most} of the bits of~$x$ correct simultaneously: $x$ and $y$ will differ in only a small fraction of the $i\in[n]$.

In Section~\ref{sec:lb} we use the previous result to recover a weaker version of Nayak's general lower bound on the number of qubits in a QRAC.
Lastly, in  Section~\ref{sec:convQRACtoRAC} we show that a QRAC can actually be converted into a classical RAC (where the states $\rho_x$ are diagonal, so they correspond to classical probability distributions), at the expense of a small increase in error probability and in message length.

\section{Preliminaries}

For matrix $A$ and $p \geq 1$, let $\| A \|_p$ denote the \emph{Schatten} $p$-norm, which is the $p$-norm of the vector of singular values of~$A$. $\| A \|_2$ is also called the Frobenius norm, denoted $\| A \|_F$. Let $\Delta(\rho , \sigma) \defeq \frac{1}{2} \|\rho - \sigma\|_1$ be the trace distance.
    We use $U(\cdot)$ to denote the uniform distribution over a set.
A two-register quantum state $\rho_{XA}$ is a classical quantum state (cq-state) if it can be written as $ \rho_{XA} = \sum_x p_x \ketbra{x}{x} \otimes \rho^x_A $
		

	


\subsection{Pretty good measurement}

We first explain the Pretty Good Measurement (PGM), also known as the square-root measurement, which was introduced in~\cite{HW:PGM,hjsww:capacity}.
Consider an ensemble $\{(P_x,\rho_x)\}_{x\in{\cal X}}$, where $m$-qubit state $\rho_x$ is given with probability $P_x$. If we measure the state with a POVM $Q=\{Q_y\}_{y\in{\cal X}}$, then the probability to correctly identify $x$ given $\rho_x$ is $\Tr(Q_x\rho_x)$, and the (average-case) success probability is $\sum_{x\in{\cal X}} P_x \Tr(Q_x\rho_x)$.
The goal is to find a POVM that  maximizes this success probability:
\[
p_{max}\defeq \max_{Q}\sum_{x\in{\cal X}} P_x \Tr(Q_x\rho_x).
\]
The optimal POVM might be very complicated and hard to find.
Fortunately, there is an easy-to-define POVM, the Pretty Good Measurement, that doesn't do much worse than the optimal POVM.
It is defined as follows: let $\rho=\sum_x P_x\rho_x$ be the average density matrix, and define the measurement operators as 
\[
Q_y \defeq P_y \rho^{-1/2}\rho_y\rho^{-1/2}
\]
(the matrix power ``${-1/2}$'' is defined only on the support of $\rho$). It is easy to check that the $Q_y$'s are positive semidefinite and their sum is identity, so this is indeed a well-defined POVM.
Let
\[
p_{PGM} \defeq \sum_{x\in{\cal X}} P_x \Tr(Q_x\rho_x)
\]
be the PGM's (average-case) success probability in identifying $x$ from $\rho_x$.
Of course, the PGM cannot do better than the best-possible POVM, so $p_{PGM}\leq p_{max}$. The beautiful property of the PGM (and the explanation for its name) is that it doesn't do much worse: $p_{PGM}\geq p_{max}^2$. This was first proved in~\cite{barnum:pgmmeasurement} (see also~\cite{montanaro:distinguishability}). 
In particular, if the optimal POVM has a success probability that's close to~1, then so will the PGM.
In fact the lower bound can be improved to~\cite[Eq.~(12)]{renes:improvedbounds} 
\[
p_{PGM}\geq p_{max}^2 + (1-p_{max})^2/(|{\cal X}|-1).
\]
If we're predicting only one bit (so $|{\cal X}|=2$) then the lower bound is 
\begin{equation}\label{eq:lowboundPpgm}
    p_{PGM}\geq p_{max}^2+(1-p_{max})^2.
\end{equation}

As an example, consider the PGM for the ensemble of the 4 states of Figure~\ref{fig:2to1QRAC}, each with probability $P_x=1/4$.
Then the average density matrix $\rho$ is just the maximally mixed state~$I/2$, and the 4 operators of the PGM are $Q_y=\frac{1}{2}\rho_y$, for $y\in\01^2$.
Since $\Tr(Q_x\rho_x)=1/2$ for each $x$, we have $p_{PGM}=1/2$. This is best-possible (and hence equal to $p_{max}$) because of Nayak's bound that was proved in Footnote~\ref{ft:Nayak}.
One can also calculate that the probability that the PGM gets $x_1$ right (i.e., that $y_1=x_1$) is $3/4$, and its probability to get $x_2$ right is also $3/4$. This $3/4$ is better than the trivial $1/2$ of random guessing, but a bit worse than the QRAC success probability of $\cos(\pi/8)^2\approx 0.85$.%
\footnote{Unsurprisingly, these two probabilities $3/4$ and $\cos(\pi/8)^2$ are exactly the optimal classical and quantum winning probabilities for the so-called CHSH game.}

\subsection{Information-theoretic preliminaries}

In this section, we introduce some basic information-theoretic definitions that we use in Section~\ref{sec:lb} and a compression theorem (Theorem~\ref{thm:max_capacity}) we use in Section~\ref{sec:convQRACtoRAC}. 


We start with the basic information-theoretic definitions. The von Neumann entropy quantifies the amount of information a quantum state holds.
\begin{definition}[von Neumann entropy]\label{def:entropy}
	 The von Neumann entropy of a quantum state $\rho$ is 
	  $$ 
   H(\rho) \defeq - \Tr(\rho\log\rho).
   $$
\end{definition}

The conditional entropy $H(A|B)$ quantifies the amount of information the $A$-register adds for a party that already holds the $B$-register.

\begin{definition}[Conditional  entropy]\label{def:cond-entropy}
	 Let  $\rho_{AB}$ be a quantum state. The conditional quantum entropy of $A$ conditioned on $B$ is 
	  $$ 
   H(A|B)_{\rho} \defeq H(\rho_{AB})-H(\rho_B).
   $$
Sometimes we omit the subscript $\rho$ when the quantum state is clear from context.
\end{definition}

	
The mutual information quantifies the amount of correlation between two registers.

\begin{definition}[Quantum mutual-information]
	\label{def:mutinfo}
	Let  $\rho_{ABC}$ be a quantum state. We define the following measures.
	$$	
 \text{Mutual-information} : \quad \I(A:B)_{\rho}\defeq  H(\rho_A) +  H(	\rho_B)
- H(\rho_{AB}).
$$
	$\text{Conditional-mutual-information} : \quad \I(A:B~|~C)_{\rho}\defeq \I(A:BC)_{\rho}-\I(A:C)_{\rho}.$\\[0.6em]
Sometimes we omit the subscript $\rho$ when the quantum state is clear from context.
\end{definition}


The following are information-theoretical definitions and results from \cite{harsha2007communication} for compressing one-way classical communication.

\begin{definition}[Max-relative entropy~\cite{Datta_2009,JRS03}] Let $\rho, \sigma$ be states with $\supp(\rho) \subseteq \supp(\sigma)$. The max-relative entropy between $\rho$ and $\sigma$ is defined as
$$ 
\dmax{\rho}{\sigma}  \defeq  \inf \{ \lambda \in \mathbb{R} : \rho \leq 2^{\lambda} \sigma  \}.
$$      
\end{definition}

\begin{definition}[Max-information~\cite{Datta_2009}]\label{def:maxinfo}
  For state $\rho_{AB}$, 
  \[
  \imax(A:B)_{\rho} \defeq   \inf_{\sigma_{B}\in \mathcal{D}(\cH_B)}\dmax{\rho_{AB}}{\rho_{A}\otimes\sigma_{B}},
  \]
  where $\mathcal{D}(\cH_B)$ denotes the set of density matrices on the $B$-register.
  If $\rho_{AB}$ is a classical state (diagonal in the computational basis) then the $\inf$ above is achieved by a classical state $\sigma_B$.
\end{definition}

If one of the registers is classical, then we can bound max information with the following facts:

\begin{fact}
\label{rhoablessthanrhoaidentity}
For a cq-state  $\rho_{XA}$ (with $X$ classical): $\rho_{XA}  \le \id_X \otimes \rho_{A}$ and   $\rho_{XA}  \le \rho_X \otimes \id_{A}$.
\end{fact}

\begin{proof}
The first inequality is proved by
    \begin{align*}
        \rho_{XA} &= \sum_x p_x \ketbra{x}{x} \otimes \rho^x_A \\
        &\leq \sum_x \ketbra{x}{x} \otimes \rho_A \\
        &= \id_X \otimes \rho_{A},
    \end{align*}
where the inequality is because $\sum_x p_x \rho^x_A = \rho_A$ and hence $p_x \rho^x_A \leq \rho_A$ for all~$x$.

The second inequality is proved by
   \begin{align*}
        \rho_{XA} &= \sum_x p_x \ketbra{x}{x} \otimes \rho^x_A \\
        &\leq \sum_x p_x \ketbra{x}{x} \otimes \id_A \\
        &= \rho_X \otimes \id_{A}.
    \end{align*}
\end{proof}

\begin{fact}[Monotonicity]\label{fact:mono}
	Let $\rho_{XA}$ be a cq-state (with $X$ classical) and $\mathcal{E}: \mathcal{L}(\cH_A)\rightarrow \mathcal{L}(\cH_B)$ be a CPTP map. Then,
	$$\mI(X:B)_{(\id_X \otimes \mathcal{E})(\rho)} \leq  \mI(X:A)_\rho  \leq \log{|A|}.$$
\end{fact}
\begin{proof}
    Let $c = \mI(X:A)_\rho$ and $\sigma_A$ be a state such that $\rho_{XA} \leq 2^c \rho_X \otimes \sigma_A$. Since $\mathcal{E}$ preserves positivity,
    $(\id_X \otimes \mathcal{E})(\rho_{XA}) \leq 2^c \rho_X \otimes \mathcal{E}(\sigma_A)$ which gives the first inequality.
    
    From Fact~\ref{rhoablessthanrhoaidentity}, 
    $\rho_{XA} \le \rho_X \otimes \id_A=2^{\log|A|}\rho_X \otimes U_A$, where $U_A$ is the uniform (maximally mixed) state on the $A$-register.  The second inequality now follows from the definition of $\mI$.
\end{proof}

The channel capacity quantifies the maximum amount of information a noisy channel can send. Analogous to channel capacity, we define max channel capacity as the maximum amount of information a noisy channel can send, measured in term of max information instead of mutual information.
\begin{definition}[Channel capacity and Max channel capacity] \label{def:max_capcity}
    Let $\cal X$ and $\cal Y$ be finite non-empty sets. Let $\mathcal{P}_\cY$ be the set of all probability distributions on $\cal Y$. A (classical) channel
with input alphabet $\cX$ and output alphabet $\cY$ is a function
$E: \cX  \rightarrow \mathcal{P}_\cY$. Let the input to the channel $E$ be a random variable $X$ with probability distribution $P_X=\{P_x\}_{x\in\cX}$ on $\cX$. Let the output of the channel be the random variable $Y$. The joint distribution of $(X,Y)$ is 
$$
P_{x,y}= P_x E(x)(y).
$$
The Shannon capacity of the channel is the mutual information between $X$ and $Y$ maximized over the choice of $P_X$:
\[
\mathcal{C}(E) \defeq \max_{P_X} I(X:Y).
\]
We define the max capacity of the channel as the max information between $X$ and $Y$ maximized over the choice of $P_X$:
\[
\mathcal{C}_{max}(E) \defeq \max_{P_X} \imax(X:Y).
\]
\end{definition}
\begin{remark}\label{rmk:cmax}
    Note that the $\imax(X:Y)$ in the definition of $\mathcal{C}_{max}$ only depends on the support of $P_X$, because for the classical $\rho_{XY}=\sum_{x} P_x \ketbra{x}{x} \otimes \rho^x_{Y}$,  $\rho_{XY} \leq 2^\lambda \rho_X \otimes \sigma_Y$ if and only if $\rho^x_{Y}\leq 2^\lambda \sigma_Y$, i.e. $\dmax{E(x)}{\sigma_Y}\leq \lambda$, for all $x$ in the support of $X$, and $\rho^x_{Y}$ is the diagonal matrix of the distribution $E(x)$, which is totally determined by the channel. 
\end{remark}


We formally define one-way protocols like \cite{harsha2007communication} as follows:
\begin{definition}[One-way protocol \cite{harsha2007communication}]\label{def:one-way}
    In a one-way protocol, the two parties Alice and Bob share a random string $R$, and also
have private random strings $R_A$ and $R_B$ respectively. Alice
receives an input $x\in \cX$. Based on the shared random string $R$
and her own private random string $R_A$, she sends a message $M(x, R, R_A)$ to Bob. On receiving the message $M$, Bob computes the output $y = y(M, R, R_B)$. The protocol is thus
specified by the two functions $M(x, R, R_A)$ and $y(M, R, R_B)$
and the distributions for the random strings $R$, $R_A$ and $R_B$. For such a protocol $\Pi$, let $\Pi(x)$ denote its (random) output
when the input given to Alice is $x$. Let $T^*_{\Pi}(x)$ be the worst-case length (measured in bits) of the message transmitted by Alice to Bob, that is,
$T^*_{\Pi}(x) = \max_{x, R, R_A}[|M (x, R, R_A)|]$. Note that the private random strings can be considered part of the shared random string if we are not concerned about minimizing the amount of shared randomness.
\end{definition}


 Combining the definitions above and the rejection sampling idea of \cite{JRS03,JRS05}, we get the following compression theorem:
 
\begin{theorem}[Channel message compression by max information]\label{thm:max_capacity}
    For every channel $E$, there is a one-way protocol (Definition~\ref{def:one-way}) $\Pi$ such that for all $x \in\cX$ and $\eta>0$, Bob’s output $\Pi(x)$ has distribution $E'(x)$ with worst-case communication length (measured in bits) of
\[
\max_{x \in \cX} T^*_{\Pi}(x) \leq \mathcal{C}_{max}(E)+O(\log\log (1/\eta))
\]
and error (measured in trace distance)
$$
\Delta(E'(x),E(x)) \leq \eta.
$$
\end{theorem}

\begin{proof}
Let $(X,Y)$ be random variables that realize the max channel capacity $\mathcal{C}_{max}(E)$. Let $Z$ be the random variable that realizes the $\inf$ in $\imax(X:Y)$ (Definition~\ref{def:maxinfo}). Then, by the definition of $\mathcal{C}_{max}$ and $\imax$, for all $x \in \cX$, $\dmax{E(x)}{Z} \leq \mathcal{C}_{max}(E)$.\footnote{As stated in Remark~\ref{rmk:cmax}, $\imax(X,Y)$ only depends on the support of $P_X$, and if there is some $x_0$ not in the support of $P_X$ with $\dmax{E(x_0)}{Z} > \mathcal{C}_{max}(E)$, then that means $P_X$ does not maximize $\imax(X,Y)$, which is a contradiction to the maximization in Definition~\ref{def:max_capcity}.} Let $a(x)=\dmax{E(x)}{Z}$. Note that ${\Pr[E(x)=z]}\leq{2^{a(x)}\Pr[Z=z]}$ for all $x \in \cX$ and $z$. If Alice and Bob pre-share iid copies $Z_1, Z_2, \dots$  of $Z$, then Alice can let Bob reject-sample $E(x)$ with error~$\eta$, as follows. Alice looks through $Z_1, Z_2, \dots$ one by one. Upon seeing $Z_i=z_i$, Alice generates a random variable $\chi_i\in \{0,1\}$ with probability
\begin{equation*}
    \Pr[\chi_i=1\mid z_i]=\frac{\Pr[E(x)=z_i]}{2^{a(x)}\Pr[Z=z_i]}.
\end{equation*}
This guarantees that
\begin{align}
        &\Pr[\chi_i=1] = \frac{1}{2^{a(x)}} \geq \frac{1}{2^{\mathcal{C}_{max}(E)}} \label{eq:pr_chi_1} \\ 
        &\Pr[z_i=z\mid\chi_i=1]=\Pr[E(x)=z]. \label{eq:rej_samp}
    \end{align}
Looking through the $i$'s, Alice sends Bob the first index $i\leq \ceil{2^{\mathcal{C}_{max}(E)}\ln(1/\eta)}$ where $\chi_i=1$, if such an $i$ exists. This requires sending at most $\mathcal{C}_{max}(E)+\log(\ln(1/\eta))+1$ bits. Upon receiving $i$, Bob outputs $z_i$, which by Eq.~\eqref{eq:rej_samp} is distributed exactly according to $E(x)$. If no $i$ was sent, then Bob gives a random output.  By Eq.~\eqref{eq:pr_chi_1}, the probability that Bob does not receive an $i$ is at most
\begin{equation}
    \left(1-\frac{1}{2^{a(x)}}\right)^{2^{\mathcal{C}_{max}(E)}\ln (1/\eta)} \leq e^{-2^{\mathcal{C}_{max}(E)}\ln(1/\eta)/2^{a(x)}}\leq e^{-\ln(1/\eta)}=\eta.
\end{equation}
Hence the total variation distance between the actual output distribution $E'(x)$ and the intended output distribution $E(x)$, is at most $\eta$.
\end{proof}





\section{Using the PGM to predict most bits of an average-case $x$}\label{sec:avcase}

Here we show how to learn most bits of $x$ in an average-case sense, where there is a known distribution $x\sim P_X$ on the encoded $n$-bit string. The idea is similar to~\cite{konig2008bounded}, as also used in~\cite{boddu2023relating}. Our main theorem is that if we use the PGM on an ensemble from a QRAC (where ${\cal X}=\01^n$), then the measurement output $y\in\01^n$ will predict each individual $x_i$ well, and hence have small Hamming distance to the encoded~$x$ (even though $y$ is very unlikely to be \emph{exactly equal} to $x$, as explained in the introduction).

\begin{theorem}\label{th:QRACmostbitsavcase}
Let $x\mapsto \rho_x$ be an $(n,m,p)$-QRAC and $P_X=\{P_x\}_{x\in\01^n}$ be a probability distribution. Suppose we apply the PGM $\{Q_y\}_{y\in\01^n}$ for the ensemble $\{(P_x,\rho_x)\}_{x\in\01^n}$ and call the measurement outcome $y$.
Then
\[
\Exp_{x\sim P_X,\, y\sim {\rm Tr}(Q_y\rho_x)}[d_H(x,y)]\leq 2p(1-p)n.
\]
\end{theorem}

\begin{proof}
Fix an $i\in[n]$.
The key to the proof is to show that $x_i=y_i$ with probability $\geq p^2+(1-p)^2$.
By ignoring (i.e., summing over) the $2^{n-1}$ possible outcomes for the $n-1$ other bits, the PGM induces a 2-outcome POVM with operators
\[
F_b\defeq \sum_{y\in\01^n:y_i=b} Q_y=\rho^{-1/2}\left(\sum_{y:y_i=b}P_y\rho_y\right)\rho^{-1/2},\, \mbox{ for }b\in\01.
\]
For $b\in\01$, define
\[
\mbox{probability }P_b\defeq \Pr_{x\sim P_X}[x_i=b]=\sum_{x:x_i=b}P_x\mbox{ ~~and~~density matrix }\rho_b\defeq \frac{1}{P_b}\sum_{x:x_i=b}P_x\rho_x
\]
(this assumes neither $P_b$ is zero; of course, identifying $x_i$ is trivial if $P_0$ or $P_1$ is~0).
Then $\rho=P_0\rho_0+P_1\rho_1$, so the average state is the same if we care about identifying the whole $x$ and if we care about identifying only $x_i$. The key observation is that the POVM $\{F_0=P_0\rho^{-1/2}\rho_0\rho^{-1/2},F_1=P_1\rho^{-1/2}\rho_1\rho^{-1/2}\}$ is exactly the PGM for the 2-state ensemble $\{(P_b,\rho_b)\}_{b\in\01}$, i.e., for identifying just the bit~$x_i$.
The QRAC property implies that there exists a POVM that identifies $x_i$ with success probability $\geq p$ even in the worst case (i.e., for every $x$), which means $p_{max}$ for this 2-state ensemble is $\geq p$.
But then by Eq.~\eqref{eq:lowboundPpgm}, the PGM $\{F_0,F_1\}$ identifies $x_i$ with success probability $\geq p_{max}^2+(1-p_{max})^2\geq p^2+(1-p)^2$; hence $\Pr[x_i\neq y_i]\leq 1-(p^2+(1-p)^2)=2p(1-p)$, where the probability is taken over $x\sim P_X$ as well as over the randomness of the measurement outcome~$y$.

The above argument works for every $i$, and the 2-outcome PGMs for the individual $x_i$'s all come from the same $2^n$-outcome PGM for $\{(P_x,\rho_x)\}$ by just ignoring $n-1$ of the bits of $y$. Therefore we can use linearity of expectation to obtain
\[
    \Exp_{x\sim P_X,\, y\sim {\rm Tr}(Q_y\rho_x)}[d_H(x,y)]=\sum_{i=1}^n \Pr[x_i\neq y_i]\leq 2p(1-p)n.\\[-1.5em]
\]
\end{proof}

If $p$ is close to~1, then the expected Hamming distance $d_H(x,y)$ is at most a small constant times $n$, and we can use Markov's inequality to show that in fact \emph{with high probability} the Hamming distance is at most a (still small but slightly bigger) constant times $n$.

\section{Using minimax theorem to predict most bits of a worst-case~$x$}\label{sec:worstcase}

The result of the previous section implies that we can learn $x$ up to small Hamming distance for a ``typical $x$'' or for ``most $x$'' (measured under the distribution $\{P_x\}$), but the PGM might actually be totally off in Hamming distance when trying to identify some of the $x$'s from their encoding state~$\rho_x$.
In this section, we will strengthen the result to a measurement that works for every $x$.

The minimax theorem from game theory (originally due to von Neumann, see e.g.~\cite{sion:minimax} for a more general version) implies the following: if ${\cal P}$ and ${\cal Q}$ are compact convex subsets of linear spaces, and $f:{\cal P}\times{\cal Q}\to\mathbb{R}$ is a linear function, then:
\begin{equation}\label{eq:minimax}
\max_{P\in{\cal P}}\min_{Q\in{\cal Q}} f(P,Q) = \min_{Q\in{\cal Q}}\max_{P\in{\cal P}} f(P,Q).
\end{equation}
We will apply the minimax theorem with ${\cal P}$ equal to the set of all probability distributions over $\01^n$, and ${\cal Q}$ equal to the set of all $2^n$-outcome POVMs on $m$ qubits.
It is easy to see that both $\cal P$ and $\cal Q$ are convex and compact (for $\cal Q$, think of an element $Q=(Q_y)_{y\in\01^n}$ as a vector of $2^n$ POVM-elements that sum to identity). Let $f$ be the expected Hamming distance between $x$ and $y$:
\[
f(P,Q)=\Exp_{x\sim P_X,\,y\sim{\rm Tr}(Q_y\rho_x)}[d_H(x,y)].
\]
Clearly $f$ is linear in $P$; to see that $f$ is also linear in $Q$, note that for a fixed $x$ we can write
\begin{align*}
\Exp_{y\sim{\rm Tr}(Q_y\rho_x)}[d_H(x,y)]
& =\sum_{i=1}^n \Pr[x_i\neq y_i] =\sum_{i=1}^n \left(\Pr[x_i=1\, \&\ y_i=0] + \Pr[x_i=0\, \&\ y_i=1]\right )\\
& =\sum_{i=1}^n \left (x_i\Tr((\sum_{y:y_i=0}Q_y)\rho_x)+(1-x_i)\Tr((\sum_{y:y_i=1}Q_y)\rho_x)\right ),
\end{align*}
which is linear in~$Q$.

Now consider the left-hand side of Eq.~\eqref{eq:minimax}. Our Theorem~\ref{th:QRACmostbitsavcase} says that for every distribution $P$, there exists a POVM~$Q$ (with outcome $y$) such that the expected Hamming distance between $x$ and $y$ is $\leq 2p(1-p)n$. Therefore $2p(1-p)n$ is an upper bound on the left-hand side of  of Eq.~\eqref{eq:minimax} and (because it's an equality) on the right-hand side as well.
Observing that the ``inner'' maximization (over $P$) on the right-hand side can be taken to be a maximization over $x$ (given $Q$, just take $P$ to be the distribution supported on an $x$ where the expected Hamming distance is largest), we have
\begin{equation}\label{eq:minimaxinstantiated}
\min_{Q\in{\cal Q}}\max_{x\in\01^n} \Exp_{y\sim {\rm Tr}(Q_y\rho_x)}[d_H(x,y)]=\min_{Q\in{\cal Q}}\max_{p\in{\cal P}} f(P,Q)\leq 2p(1-p)n.
\end{equation}
The POVM $Q$ that minimizes the left-hand side of Eq.~\eqref{eq:minimaxinstantiated} gives us the strengthening of Theorem~\ref{th:QRACmostbitsavcase}:

\begin{theorem}\label{th:QRACmostbitsworstcase}
Let $x\mapsto \rho_x$ be an $(n,m,p)$-QRAC.
There exists a POVM $\{Q_y\}_{y\in\01^n}$ such that for every $x\in\01^n$:
\[
\Exp_{y\sim{\rm Tr}(Q_y\rho_x)}[d_H(x,y)]\leq 2p(1-p)n,
\]
where the expectation is taken over the distribution of the measurement outcome~$y$ (not over~$x$).
\end{theorem}

\section{Lower bound for the number of qubits of QRACs}\label{sec:lb}

Our result that one can get $x$ up to small Hamming distance from $\rho_x$ implies a lower bound on the number of qubits in $\rho_x$, Theorem~\ref{th:lowerboundonm}. Our lower bound is quantitatively weaker than Nayak's $m\geq (1-H(p))n$.


\begin{theorem}\label{th:lowerboundonm}
Let $x\mapsto \rho_x$ be an $(n,m,p)$-QRAC. Then $m\geq ( 1-H(2p(1-p)) )n - \log_2(n+1)$.
\end{theorem}

Note that the factor $1-H(2p(1-p))$ in our lower bound is a constant factor worse than Nayak's optimal constant $1-H(p)$ if $p$ is close to~1, but quadratically worse if $p$ is close to~$1/2$.\footnote{If $p=1/2+\beta$, then $1-H(p)=\Theta(\beta^2)$, while $1-H(2p(1-p))=\Theta(\beta^4)$ because $2p(1-p)=1/2-2\beta^2$.}

\bigskip

\begin{proof}
Let $X\in\01^n$ be uniformly random, and $Y\in\01^n$ the random variable which is the output of the PGM on the ensemble $\{(1/2^n,\rho_x)\}$. Let $D=d_H(X,Y)$ be a new random variable, which ranges over $0,1,\ldots,n$. We have $\Exp[D]\leq 2p(1-p)n$ by Theorem~\ref{th:QRACmostbitsworstcase}.
 If we condition on $D=d$, then given $Y$ there are only $\binom{n}{d}\leq 2^{H(d/n)n}$ possibilities left for~$X$, so the conditional entropy of $X$ is $H(X|Y,D=d)\leq H(d/n)n$.
Taking expectation over the different values of $D$, we have
\[
H(X|Y,D)=\Exp_d[H(X|Y,D=d)]\leq \Exp_d[H(d/n)n]\leq H(\Exp_d[d/n])\,n\leq  H(2p(1-p))n
\]
where the second inequality is Jensen's inequality ($H(\cdot)$ is a concave function), and the last inequality uses that $H(q)$ is increasing for $q\in[0,1/2]$. We have $H(X|Y,D)=H(X,D|Y)-H(D|Y)\geq H(X|Y)-\log_2(n+1)$,
and hence we can lower bound the mutual information between $X$ and $Y$ as
\[
I(X:Y)=H(X)-H(X|Y)\geq H(X)-H(X|Y,D)-\log_2(n+1)\geq (1-H(2p(1-p)))n - \log_2(n+1).
\]
By the data-processing inequality and the fact that $(X,\rho_X)$ is a classical-quantum state, we have the upper bound $I(X:Y)\leq I(X:\rho_X)\leq H(\rho_X)\leq m$. The theorem now follows by combining our upper and lower bounds on $I(X:Y)$.
\end{proof}

\section{Converting QRAC to RAC}\label{sec:convQRACtoRAC}

In this section, we give a constructive version of the result of Section~\ref{sec:lb} by ``converting'' a QRAC to a classical random access code (RAC) directly, using the compression scheme of~\cite{JRS03,harsha2007communication} and symmetrization ideas of~\cite{ALMO:QRAC}.

\begin{theorem}\label{th:QRACtoRAC}
     For all $\eta>0$, if an $(n,m,p)$-QRAC exists, then an $(n,(m+O(\log n + \log (1/\eta))),1-2p(1-p)-\eta)$-RAC exists. 
\end{theorem}

\begin{proof}
Let  be $x\mapsto\rho_x$ be the encoding map of the $(n,m,p)$-QRAC.

By Theorem~\ref{th:QRACmostbitsavcase},
for every distribution $P_X$, the PGM $\{Q_y\}_{y\in\01^n}$ for the ensemble $\{(P_x,\rho_x)\}_{x\in\01^n}$ satisfies
\[
\Exp_{x\sim P_X,\, y\sim {\rm Tr}(Q_y\rho_x)}[d_H(x,y)]\leq 2p(1-p)n.
\]


Consider the following ``symmetrized'' encoding-decoding protocol using shared randomness.
\newcommand{\sd}{\text{SHIFT}_d}

\SetAlgorithmName{Protocol}{Protocol}

\begin{algorithm}
    \caption{Symmetrized QRAC\label{prot:sym}}
\SetKwInOut{Enc}{Encoding}
\SetKwInOut{Dec}{Decoding}
\SetKwInOut{Init}{Initializing}
\Init{use shared randomness to sample a uniformly random shift $d\in [n]$ and a uniformly random string $r\in \01^n$. Define the shift function as 
     $\sd: \01^n\rightarrow \01^n, (\sd(x))_i =x_{((i+d-1 {\rm ~mod~} n)+1)}$}
\Enc{calculate $x'=\sd(x\oplus r)$. Encode $x$ into $\overline{\rho}_{x}\defeq \rho_{x'}$.}
\Dec{measure with the PGM $\{Q_y\}$ in Theorem~\ref{th:QRACmostbitsavcase} with $P_X=U(\01^n)$ and record the result as $y'$. Output $y=\text{SHIFT}_{-d}(y')\oplus r$. }
\end{algorithm}

For each $x\in\01^n$ and $i\in[n]$, the probability of $y_i$ correctly reproducing $x_i$ in Protocol~\ref{prot:sym} is  
\begin{align}\label{eq:symprotocol}
    \Pr[x_i=y_i]& = \Pr[x'_{(i-d)}=y'_{(i-d)}] \nn \\ 
    &=\Pr[\E_{x'\sim U(\01^n)}\E_{i\sim U([n])}\E_{y'\sim{\rm Tr}(Q_{y'}\rho_{x'})}[x'_i=y'_i]] \nn \\ 
    &=\Exp_{x'\sim U(\01^n),\, y'\sim {\rm Tr}(Q_{y'}\rho_{x'})}\left[1-\frac{1}{n}d_H(x',y')\right]\geq 1-2p(1-p), 
\end{align}
where in the second equality we used that both $x'$ and $i-d$ are uniformly random, and in the inequality we used Theorem~\ref{th:QRACmostbitsavcase} with $P_X=U(\01^n)$. 




  Denote the effective encoding map and decoding measurement in Protocol~\ref{prot:sym} as $x\mapsto\overline{\rho}_x$ and $\{\overline{Q}_y\}$, respectively. 
  These depend on the same shared random $r,d$. Consider the RAC where the encoded message is just the random variable corresponding to  the classical string $y$ in Protocol~\ref{prot:sym}. By Equation~\eqref{eq:symprotocol}, this RAC can be decoded with success probability $1-2p(1-p)$ by guessing $x_i=y_i$.
 This RAC has bad encoded length since the length of $y$ is the same as the length of $x$, namely $n$ bits. However, noting that $y$ comes from measuring an $m$-qubit quantum state, we know it has ``low correlation'' with $x$. More specifically, consider the channel $E$ that produces $y$ by measuring $\overline{\rho}_{x}$ with the effective measurement $\overline{Q}$ of Protocol~\ref{prot:sym}:
$$
E(x)(y)= \Tr (\overline{Q}_y \overline{\rho}_x).
$$
Then we can upper bound the max channel capacity of $E$ by Fact~\ref{fact:mono}: 
$$
\mathcal{C}_{max}(E) = \max_{P_X=\{P_x\}} \imax(X:Y)  \leq \max_{P_X=\{P_x\}} \imax(X:\overline{\rho})  \leq m,
$$
where in the third expression, $X$ and $\overline\rho$ are the two registers of the state $\sum_x P_x\ketbra{x}{x} \otimes \overline{\rho}_x$. Therefore, we can use the compression protocol of Theorem~\ref{thm:max_capacity} to compress $y$ into another short message $\tilde{y}$ with  worst-case length $m+O(\log\log (1/\eta))$ using shared randomness between the encoder and the decoder, with an $\eta/2$ loss in the success probability.

Lastly, we will  use the probabilistic method to get rid of the shared randomness, at the expense of another loss of at most $\eta/2$ in the success probability, and $O(\log(n)+\log(1/\eta))$ additional bits in the length of the message.%
\footnote{This is analogous to the proof of Newman's theorem that shows how to get rid of shared randomness in communication complexity~\cite{newman:random,kushilevitz&nisan:cc}.}
For $s=(r,d)$, $x\in\01^n$ and $i\in[n]$, let $P_{sxi}$ denote the probability that $y_i\neq x_i$ in the protocol, for those $r,d,x,i$.
The error probability of the protocol for decoding $x_i$ from this particular $x$, is $\Exp_{s\sim U}[P_{sxi}]$, where $U$ is the uniform distribution over all possible $s=(r,d)$.

Pick a set $S$ (whose size will be determined soon) of $s$'s uniformly at random. For each pair $x,i$, let $B_{xi}$ denote the ``bad event'' that sampling an $s$ uniformly from the set $S$ (rather than from the uniform distribution $U$) increases the error probability of the protocol by more than $\eta/2$ for this particular $x,i$:
\[
B_{xi}:~~~~ 
\Exp_{s\sim S}[P_{sxi}]> \Exp_{s\sim U}[P_{sxi}]+\eta/2.
\] 
Note that $\Exp_{s\sim S}[P_{sxi}]$ is the average over $|S|$ many iid random variables in $[0,1]$, each of expectation $\Exp_{s\sim U}[P_{sxi}]$. By a Chernoff bound, the probability (over the choice of $S$) of the above event $B_{xi}$ is then exponentially small in $\eta^2|S|$.
If we choose $|S|=O(n/\eta^2)$ with a sufficiently large constant in the $O(\cdot)$, then $\Pr[B_{xi}]<1/(2^n n)$.
By a union bound over all $2^n n$ pairs $(x,i)$, we have $\Pr[\exists (x,i)~s.t.~B_{xi}]<1$. Hence there \emph{exists} a choice of set $S$ such that $B_{xi}$ is false \emph{for all pairs $(x,i)$ simultaneously}.
Fix such an~$S$.
We now modify the message by sampling an $s\in S$ and appending to it the compressed random variable $\tilde{y}$ that the protocol produces with this choice of $s=(r,d)$. 
This modification of the protocol increases the  error probability by at most $\eta/2$, because $B_{xi}$ is false for all $x,i$ thanks to our choice of $S$. 
In total, the error probability compared to the original QRAC went up by at most $\eta/2+\eta/2=\eta$. We do not need shared randomness between encoder and decoder anymore, because the random variable $s$ is now part of the message.
It increases the length of the message by the number of bits needed to write down~$s$ (knowing that it comes from the fixed set $S$), which is $\ceil{\log_2|S|}=O(\log(n)+\log(1/\eta))$ bits.
\end{proof}






\section{Open problem}

An interesting open problem related to the ``obfuscation'' of QRACs is the possibility to strengthen Theorem~\ref{th:QRACmostbitsworstcase} to get a $y$ such that $y_i=x_i$ with error probability $2p(1-p)$ for every QRAC, every $x$, and every $i$. The motivation is that the prediction of $x$ we get from Theorem~\ref{th:QRACmostbitsworstcase} is weaker than what we can get from a RAC with error probability $2p(1-p)$. In particular, although the strings $y$ and $x$ have low Hamming distance to each other, there may exist a particular index $i$ such that $x_i$ never equals~$y_i$, while in a proper RAC, it holds \emph{for every $i$} that $x_i=y_i$ with good probability. 



\begin{conjecture}\label{con:QRACworst_y_open}
    Let $x\mapsto \rho_x$ be an $(n,m,p)$-QRAC.
There exists a POVM $\{Q_y\}_{y\in\01^n}$ such that for every $x\in\01^n$ and $i\in [n]$:
\[
\Exp_{y\sim{\rm Tr}(Q_y\rho_x)}\left[\Pr[x_i \neq y_i]\right]\leq 2p(1-p),
\]
where the expectation is taken over the distribution of the measurement outcome~$y$.
\end{conjecture}

Note that although  Theorem~\ref{th:QRACmostbitsavcase} actually works for every $i$, one cannot directly apply the minimax theorem as in Theorem~\ref{th:QRACmostbitsworstcase} to prove Conjecture~\ref{con:QRACworst_y_open}, because the lowest success probability over all $i$ is not a linear function in $P$ and $Q$.




\paragraph{Acknowledgments.}
We thank Rahul Jain for helpful comments on a draft of this paper.

\bibliographystyle{alpha}
\bibliography{qc}

\end{document}